\ifx\ESAVer\undefined
   \newcommand{\InESAVer}[1]{}%
\else
   \def\SubmitVer{TRUE}
   \newcommand{\InESAVer}[1]{#1}%
   \newcommand{\InNotESAVer}[1]{}%
\fi

\ifx\SubmitVer\undefined
   \newcommand{\InSubmitVer}[1]{}%
   \newcommand{\InNotSubmitVer}[1]{#1}%
   \providecommand{\InESAVer}[1]{}%
   \renewcommand{\InESAVer}[1]{}%
   \newcommand{\InNotESAVer}[1]{#1}%
\else
   \def\ProofInAppendix{TRUE}
   \newcommand{\InSubmitVer}[1]{#1}%
   \newcommand{\InNotSubmitVer}[1]{}%
\fi

\InSubmitVer{%
   \InNotESAVer{
      \documentclass[11pt]{article}%
   }
}
\InESAVer{
      \documentclass{llncs}%
}
\InNotSubmitVer{%
   \documentclass[12pt]{article}%
}

\usepackage{gensymb}
\usepackage{graphicx}
\usepackage{euscript}%
\usepackage{xspace}%
\usepackage{color}%
\usepackage{amsmath}%

\usepackage{amssymb}%
\usepackage{paralist}%
\usepackage{styles/picins}%
\usepackage{enumerate}

\InSubmitVer{%
   \InNotESAVer{\usepackage{lineno}}%
}

\usepackage{hyperref}%
\hypersetup{%
   breaklinks,%
   colorlinks=true,%
   linkcolor=[rgb]{0.45,0.0,0.0},%
   citecolor=[rgb]{0,0,0.45}
}

\InNotESAVer{%
   \InSubmitVer{%
      \usepackage[in]{fullpage}%
   }%
   \InNotSubmitVer{%
      \usepackage[cm]{fullpage}%
   }%
}

\usepackage{salgorithm}%

\usepackage{proofapnd}

\InNotESAVer{%
   \usepackage{titlesec}%
   \titlelabel{\thetitle. }%

   \usepackage[amsmath,amsthm,thmmarks]{ntheorem}
   \theoremseparator{.}
}

\newcommand{\eps}{\varepsilon}%

\newcommand{\funcA}{f}

\newcommand{\ArrC}{\mathcal{A}}
\newcommand{\ArrX}[1]{\mathcal{A}\pth{#1}}

\newcommand{\cardin}[1]{\left| {#1} \right|}
\newcommand{\AreaOverlap}[2]{\mu\pth{#1, #2}}

\newcommand{\aFunc}{\psi}

\newcommand{\AOA}[1]{\psi\pth{#1}}

\newcommand{\ApproxOverlapC}{\mu'}

\newcommand{\maxOverlap}[2]{\mu_{\max}\pth{#1, #2}}
\newcommand{\areaX}[1]{\mathrm{area}\pth{#1}}
\newcommand{\widthX}[1]{\mathrm{width}\pth{#1}}
\newcommand{\diameterX}[1]{\mathrm{diameter}\pth{#1}}
\newcommand{\diskX}[1]{\mathrm{disk}\pth{#1}}

\newcommand{\SubhroThanks}[1]{\thanks{%
      Department of Computer Science;
      University of Illinois; %
      201 N. Goodwin Avenue; %
      Urbana, IL, 61801, USA; %
      {\tt \si{sroy}9\atgen{}\si{illinois}.\si{edu}}; %
      {\tt \url{http://web.engr.illinois.edu/\string~sroy9/}}%
      . %
      #1}}

\newcommand{\atgen}{\symbol{'100}}
\newcommand{\SarielThanks}[1]{\thanks{%
      Department of
      Computer Science; University of Illinois; 201 N. Goodwin Avenue;
      Urbana, IL, 61801, USA; {\tt
         \si{sariel}\atgen{}\si{uiuc.edu}}; %
      {\tt \url{http://sarielhp.org}.} #1}}

\definecolor{blue25}{rgb}{0,0,0.55}%
\newcommand{\emphic}[2]{%
   \textcolor{blue25}{%
      \textbf{\emph{#1}}}%
   \index{#2}}
\newcommand{\emphi}[1]{\emphic{#1}{#1}}
\newcommand{\pth}[2][\!]{#1\left({#2}\right)}

\newcommand{\Frechet}{Fr\'{e}c{h}e{}t\xspace}%
\newcommand{\remove}[1]{}
\newcommand{\etal}{\textit{et~al.}\xspace}

\providecommand{\si}[1]{#1}

\numberwithin{figure}{section}

\InNotESAVer{%
   \newtheorem{theorem}{Theorem}[section]
   \newtheorem{lemma}[theorem]{Lemma}%

   \newtheorem{defn}[theorem]{Definition}
   
   \newtheorem{algorithm}[theorem]{Algorithm}%

{\theorembodyfont{\rm}%
   \newtheorem{observation}[theorem]{Observation}%
   \newtheorem{definition}[theorem]{Definition}%
   \newtheorem{problem}[theorem]{Problem}%
}

}

\InESAVer{%
   \newtheorem{observation}[theorem]{Observation}%
   \newtheorem{defn}[theorem]{Definition}
   \newtheorem{algorithm}[theorem]{Algorithm}%
}

\newcommand{\WidthX}[1]{\omega \pth{#1}}

\newcommand{\apndlab}[1]{\label{apnd:#1}}
\newcommand{\apndref}[1]{Appendix~\ref{apnd:#1}}

\newcommand{\obslab}[1]{\label{observation:#1}}
\newcommand{\obsref}[1]{Observation~\ref{observation:#1}}

\newcommand{\constRect}{c_r}

\newcommand{\constC}{c_3}

\newcommand{\constD}{c_4}

\newcommand{\alglab}[1]{\label{alg:#1}}
\newcommand{\algref}[1]{Algorithm~\ref{alg:#1}}

\newcommand{\thmlab}[1]{{\label{theo:#1}}}

\newcommand{\lemlab}[1]{\label{lemma:#1}}
\newcommand{\lemref}[1]{Lemma~\ref{lemma:#1}}

\newcommand{\seclab}[1]{\label{sec:#1}}
\newcommand{\secref}[1]{Section~\ref{sec:#1}}

\newcommand{\figlab}[1]{\label{figure:#1}}
\newcommand{\figref}[1]{Figure~\ref{figure:#1}}

\newcommand{\Ellipse}{\mathcal{E}}

\newcommand{\MakeBig} {\rule[-.2cm]{0cm}{0.4cm}}
\renewcommand{\th}{\si{th}\xspace}

\newcommand{\Trans}{{\mathcal{T}}}

\newcommand{\trans}{{\mathsf{t}}}

\newcommand{\transY}{{\mathsf{t}_\bodyY}}

\newcommand{\PntSet}{\mathsf{S}}

\renewcommand{\Re}{{\rm I\!\hspace{-0.025em} R}}

\newcommand{\brc}[1]{\left\{ {#1} \right\}}
\newcommand{\sep}[1]{\,\left|\, {#1} \MakeBig\right.}
\newcommand{\scaleSimX}[2]{\mathrm{ssim}\pth{#1, #2}}

\newcommand{\Line}{\ell}

\newcommand{\pnt}{\mathsf{p}}%
\newcommand{\pntA}{\mathsf{q}}%

\newcommand{\vecA}{\vec{u}}
\newcommand{\vecB}{\vec{v}}

\newcommand{\vecC}{\vec{z}}

\newcommand{\vecD}{\overrightarrow{w}}
\newcommand{\vecMax}{\overrightarrow{m}}

\newcommand{\vecG}{\vec{g}}
\newcommand{\vecX}{\vec{x}}
\newcommand{\vecXA}{\vec{x_1}}
\newcommand{\vecY}{\vec{y}}

\newcommand{\elpsX}{\Ellipse_\bodyX}
\newcommand{\elpsY}{\Ellipse_\bodyY}
\newcommand{\elpsG}{\Ellipse_\bodyG}
\newcommand{\elpsB}{\Ellipse_\bodyB}

\newcommand{\Poly}{\mathsf{P}}%
\newcommand{\PolyA}{\mathsf{Q}}%

\newcommand{\bodyB}{\mathsf{B}}%

\newcommand{\bodyX}{\mathsf{X}}%
\newcommand{\bodyY}{\mathsf{Y}}%
\newcommand{\bodyXA}{\mathsf{X'}}%
\newcommand{\bodyYA}{\mathsf{Y'}}%

\newcommand{\bodyXT}{\mathsf{X}_\Trans}%
\newcommand{\bodyYT}{\mathsf{Y}_\Trans}%
\newcommand{\bodyXTA}{\mathsf{X}_\Trans'}%
\newcommand{\bodyYTA}{\mathsf{Y}_\Trans'}%

\newcommand{\bodyMax}{\mathsf{M}}%

\newcommand{\bodyC}{\mathsf{Z}}%

\newcommand{\bodyD}{\mathsf{D}}%
\newcommand{\bodyG}{\mathsf{G}}%

\newcommand{\transform}{\mathsf{M}}

\providecommand{\ds}{\displaystyle}

\newcommand{\bbox}{B}

\newcommand{\ceil}[1]{\left\lceil {#1} \right\rceil}

\newcommand{\LevelSetX}[2]{L_{#1}\pth{#2}}

\providecommand{\Arr}{\mathcal{A}}

\newcommand{\epsA}{\epsilon}

\newcommand{\rectA}{\mathsf{r}}

\newcommand{\rectM}{\mathsf{r}_\bodyMax}
\newcommand{\rectX}{\mathsf{r}_\bodyX}
\newcommand{\rectY}{\mathsf{r}_\bodyY}

\newcommand{\rectC}{\mathsf{r}_\bodyC}

\newlength{\savedparindent}
\newcommand{\SaveIndent}{\setlength{\savedparindent}{\parindent}}
\newcommand{\RestoreIndent}{\setlength{\parindent}{\savedparindent}}

\newcommand{\algSSim}{\Algorithm{ssim}\xspace}
\newcommand{\constApproxByRect}{\Algorithm{constApproxByRect}\xspace}
\newcommand{\approxPolyWidth}{\Algorithm{approxPolygon}\xspace}

\newcommand{\approxLevelSet}{\Algorithm{approxLevelSet}\xspace}
\newcommand{\wdX}{\omega_\bodyX}
\newcommand{\wdY}{\omega_\bodyY}

\newcommand{\numPolygons}{\nu}
\newcommand{\totalN}{\rho}

\newcommand{\slice}{\mathsf{S}}
\newcommand{\TSUMHard}{\textsf{3SUM}-Hard\xspace}

\newcommand{\errBX}{\mathcal{E}_\bodyX}%
\newcommand{\errBY}{\mathcal{E}_\bodyY}%
\newcommand{\err}{\mathcal{E}}

\newcommand{\floor}[1]{\left\lfloor {#1} \right\rfloor}

\newcommand{\mysubsubsection}[1]{\subsubsection{#1.}}

\begin{document}
\InNotESAVer{\InSubmitVer{\linenumbers}}%

\title{Approximating the Maximum Overlap of Polygons under
   Translation%
   \InNotESAVer{%
      \footnote{Work on this paper was partially supported by NSF AF
         awards CCF-0915984 and CCF-1217462.}%
   }
}

\author{%
   Sariel Har-Peled%
   \InNotESAVer{\SarielThanks{}}%
   \InESAVer{\inst{1}}%
   \and%
   Subhro Roy%
   \InNotESAVer{\SubhroThanks{}}%
   \InESAVer{\inst{1}}%
}

\InESAVer{%
   \institute{%
      University of Illinois, Urbana-Champaign%
   }%
}

\date{\today}

\maketitle

\begin{abstract}
    Let $\Poly$ and $\PolyA$ be two simple polygons in the plane of
    total complexity $n$, each of which can be decomposed into at most
    $k$ convex parts. We present an $(1-\eps)$-approximation
    algorithm, for finding the translation of $\PolyA$, which
    maximizes its area of overlap with $\Poly$.  Our algorithm runs in
    $O\pth{ c n }$ time, where $c$ is a constant that depends only on
    $k$ and $\eps$.

    This suggest that for polygons that are ``close'' to being convex,
    the problem can be solved (approximately), in near linear time.
\end{abstract}

\InNotESAVer{%
   \InSubmitVer{%
      \thispagestyle{empty}%
      \newpage%
      \setcounter{page}{1}%
   }%
}

\section{Introduction}

Shape matching is an important problem in databases, robotics,
visualization and many other fields.  Given two shapes, we want to
find how similar (or dissimilar) they are.  Typical problems include
matching point sets by the Hausdorff distance metric, or matching
polygons by the Hausdorff or \Frechet distance between their
boundaries. See the survey by Alt and Guibas \cite{ag-dgsmi-00}.

The maximum area of overlap is one possible measure for shape matching
that is not significantly effected by noise.  Mount \etal
\cite{msw-aotp-96} studied the behavior of the area of overlap
function, when one simple polygon is translated over another simple
polygon. They showed that the function is continuous and piece-wise
polynomial of degree at most two. If the polygons $\Poly$ and $\PolyA$
have complexity $m$ and $n$, respectively, the area of overlap
function can have complexity of $\Theta(m^2n^2)$. Known algorithms to
find the maximum of the function work by constructing the entire
overlap function. It is also known that the problem is \TSUMHard
\cite{bh-pctmh-01}, that is, it is believed no subquadratic time
algorithm is possible for the problem.

\paragraph{Approximating maximum overlap of general polygons.}

Cheong \etal \cite{ceh-fgsms-07} gave a $(1 - \eps)$-approximation
algorithm for maximizing the area of overlap under translation of one
simple polygon over the other using random sampling
techniques. However, the error associated with the algorithm is
additive, and the algorithm runs in near quadratic time. Specifically,
the error is an $\eps$ fraction of the area of the smaller of the two
polygons. Under rigid motions, the running time deteriorates to being
near cubic.  More recently, Cheng and Lam \cite{cl-smurm-13} improved
the running times, and can also handle rigid motions, and present a
near linear time approximation algorithm if one of the polygons is
convex.

\paragraph{Maximum overlap in the convex case under translations.}

\si{de Berg} \etal \cite{bcdkt-cmotc-98} showed that finding maximum
overlap translation is relatively easier in case of convex
polygons. Specifically, the overlap function in this case is unimodal
(as a consequence of the Brunn-Minkowski Theorem). Using this
property, they gave a near linear time exact algorithm for computing
the translation that maximizes the area of overlap of two convex
polygons. The complexity of the graph of the overlap function is only
$O\pth{ m^2 + n^2 + \min( m^2n, mn^2)}$ in this case. Alt \etal
\cite{afrw-mcsrs-98} gave a constant-factor approximation for the
minimum area of the symmetric difference of two convex
polygons. 

\paragraph{Approximating maximum overlap in the convex case.}
As for $(1-\eps)$-approximation, assuming that the two polygons are
provided in an appropriate form (i.e., the vertices are in an array in
their order along the boundary of the polygon), then one can get a
sub-linear time approximation algorithm.  Specifically, Ahn \etal
\cite{acpsv-motpc-07} show an $(1-\eps)$-approximation algorithm, with
running time $O((1/\eps)\log (n/\eps))$ for the case of translation,
and $O((1/\eps)\log n + (1/\eps)^2 \log 1/\eps))$ for the case of
rigid motions.  (For a result using similar ideas in higher dimensions
see the work by Chazelle \etal \cite{clm-sga-05}.)

\paragraph{Overlap of union of balls.}

\si{de Berg} \etal \cite{bcgkov-maotu-04} considered the case where
$X$ and $Y$ are disjoint unions of $m$ and $n$ unit disks, with $m
\leq n$. They computed a $(1 - \eps)$ approximation for the maximal
area of overlap of $X$ and $Y$ under translations in time
$O((nm/\eps^2) \log(n/\eps))$.  Cheong \etal \cite{ceh-fgsms-07} gave
an additive error $\eps$-approximation algorithm for this case, with
near linear running time.

\paragraph{Other relevant results.}
Avis \etal \cite{abtszs-osacp-96} computes the overlap of a polytope
and a translated hyperplane in linear time, if the polytope is
represented by a lattice of its faces. Vigneron \cite{v-gosaf-14}
presented $(1-\eps)$-approximation algorithms for maximum overlap of
polyhedra (in constant dimension) that runs in polynomial time.  Ahn
\etal \cite{acky-ocprm-14} approximates the maximum overlap of two
convex polytopes in three dimensions under rigid motions.  Ahn \etal
\cite{acr-mocpt-13} approximates the maximum overlap of two polytopes
in $\Re^d$ under translation in $O\pth{n^{\floor{d/2}+1} \log^d n}$
time.

\section*{Our results}

As the above indicates, there is a big gap between the algorithms known
for the convex and non-convex case. Our work aims to bridge this gap,
showing that for ``close'' to convex polygons, under translation, the
problem can be solved approximately in near linear time.

Specifically, assume we are given two polygons $\Poly$ and $\PolyA$ of
total complexity $n$, such that they can be decomposed into $k$ convex
parts, we show that one can $(1-\eps)$-approximate the translation of
$\PolyA$, which maximizes its area of overlap with $\Poly$, in linear
time (for $k$ and $\eps$ constants). The translation returned has
overlap area which is at least $(1-\eps) \maxOverlap{\Poly}{\PolyA}$,
where $\maxOverlap{\Poly}{\PolyA}$ is the maximum area of overlap of
the given polygons.

\paragraph{Approach.} 

We break the two polygons into a minimum number of convex parts. We
then approximate the overlap function for each pair of pieces
(everywhere). This is required as one cannot just approximate the two
polygons (as done by Ahn \etal \cite{acpsv-motpc-07}) since the
optimal solution does not realize the maximum overlap of each pair of
parts separately, and the alignment of each pair of parts might be
arbitrary.

To this end, if the two convex parts are of completely different
sizes, we approximate the smaller part, and approximate the overlap
function by taking slices (i.e., level sets) of the overlap
function. In the other case, where the two parts are ``large'', which
is intuitively easier, we can approximate both convex parts, and then
the overlap function has constant complexity. Finally, we overlap all
these functions together, argue that the overlap has low complexity,
and find the maximum area of overlap.

Our approach has some overlap in ideas with the work of Ahn \etal
\cite{acpsv-motpc-07}. In particular, a similar distinction between
large and small overlap, as done in \secref{small} and \secref{large}
was already done in \cite[Theorem 17]{acpsv-motpc-07}.

\paragraph{Why the ``naive'' solution fails?}

The naive solution to our problem is to break the two polygons into
$k$ convex polygons, and then apply to each pair of them the
approximation of Ahn \etal \cite{acpsv-motpc-07}. Now, just treat the
input polygon as the union of their respective approximations, and
solve problem using brute force approach. This fails miserably as the
approximation of Ahn \etal \cite{acpsv-motpc-07} captures only the
maximum overlap of the two polygons. It does not, and can not,
approximates the overlap if two convex polygons are translated such
that their overlap is ``far'' from the maximum configuration,
especially if the two polygons are of different sizes. This issue is
demonstrated in more detail in the beginning of \secref{small}.
A more detailed counterexample is presented in \apndref{counter:example}.

\paragraph{Paper organization.}

We start in \secref{prelim} by defining formally the problem, and
review some needed results. In \secref{building}, we build some
necessary tools. Specifically, we start in \secref{better:approx} by
observing that one can get $O(1/\eps)$ approximation of a convex
polygon, where the error is an $\eps$-fraction of the width of the
polygon. In \secref{level:set}, we show how to compute a level set of
the overlap function of two convex polygons efficiently.  In
\secref{shape:overlap}, we show that, surprisingly, the polygon formed
by the maximum overlap of two convex polygons, contains (up to scaling
by a small constant and translation) the intersection of any
translation of these two convex polygons. Among other things this
implies an easy linear time constant factor approximation for the
maximum overlap (which also follows, of course, by the result of Ahn
\etal \cite{acpsv-motpc-07}).  In \secref{approx:overlap:convex}, we
present the technical main contribution of this paper, showing how to
approximate, by a compact representation that has roughly linear
complexity, the area overlap function of two convex polygons.  In
\secref{approx:main} we put everything together and present our
approximation algorithm for the non-convex case.

\section{Preliminaries}%
\seclab{prelim}

For any vector $\trans \in \Re^2$ and a set $\PolyA$, let $\trans +
\PolyA$ denote the translation of $\PolyA$ by $\trans$; formally,
$\trans + \PolyA = \brc{ \trans + \pntA \sep{ \pntA \in \PolyA}}$.
Also let $\AreaOverlap{\Poly}{\PolyA} = \areaX{\Poly \cap \PolyA}$,
which is the \emphi{area of overlap} of sets $\Poly$ and $\PolyA$.  We
are interested in the following problem.

\begin{problem}
    We are given two polygons $\bodyX$ and $\bodyY$ in the plane, such
    that each can be decomposed into at most $k$ convex polygons. The
    task is to compute the translation $\trans$ of $\bodyY$, which
    maximizes the area of overlap between $\bodyX$ and $\trans +
    \bodyY$. Specifically our purpose is to approximate the quantity
    \begin{align*}
	\maxOverlap{\bodyX}{\bodyY} = \max_{\trans \in \Re^2}\;
        \AreaOverlap{\bodyX}{ \trans + \bodyY }.
    \end{align*}    
\end{problem}


For a polygon $\Poly$, let $\cardin{\Poly}$ denote the number of
vertices of $\Poly$. For $\bodyX, \bodyY \subseteq \Re^d$, the set
$\bodyX$ is \emphi{contained under translation} in $\bodyY$, denoted
by $\bodyX \sqsubseteq \bodyY$, if there exists $\vecX$ such that
$\vecX + \bodyX \subseteq \bodyY$.

\paragraph{Unimodal.}
A function $f:\Re\rightarrow \Re$ is \emphi{unimodal}, if there is a
value $\alpha$, such that $f$ is monotonically increasing (formally,
non-decreasing) in the range $[-\infty, \alpha]$, and $f$ is
monotonically decreasing (formally, non-increasing) in the interval
$[\alpha, +\infty]$.

\paragraph{From width to inner radius.}
For a convex polygon $\Poly$, the \emphi{width} of $\Poly$, denoted by
$\WidthX{\Poly}$, is the minimum distance between two parallel lines
that enclose $\Poly$.

\begin{lemma}[\cite{gk-iojrc-92}]%
    \lemlab{width:to:ball}%
    For a convex shape $\bodyX$ in the plane, we have that the largest
    disk enclosed inside $\bodyX$, has radius at least
    $\widthX{\bodyX}/2\sqrt{3}$.
\end{lemma}

\paragraph{Convex Decomposition of Simple Polygons.}

A vertex of a polygon is a \emphi{notch} if the internal angle at this
vertex is reflex (i.e. $> 180 \degree$).  For a non-convex polygon
$\Poly$ with $n$ vertices and $r$ notches, Keil and Snoeyink
\cite{ks-otbcd-02} solves the minimal convex decomposition problem in
$O\pth{n + r^2\min( r^2, n)}$ time, that is, they compute a
decomposition of $\Poly$ into minimum number of convex
polygons. Observe, that if the number of components in the minimum
convex decomposition is $k$, the number of notches $r$ is upper
bounded by $2k$.

\paragraph{Scaling similarity between polygons.}

For two convex polygons $\bodyX$ and $\bodyY$, let us define their
\emphi{scaling similarity}, denoted by $\scaleSimX{\bodyX}{\bodyY}$,
as the minimum number $\alpha \geq 0$, such that $\bodyX \sqsubseteq
\alpha \bodyY$ .  Using low-dimensional linear programming, one can
compute $\scaleSimX{\bodyX}{\bodyY}$ in linear time. In particular,
the work by Sharir and Toledo \cite{st-epcp-94} implies the following.

\begin{lemma}[\algSSim]%
    \lemlab{scaling}%
    Given two convex polygons $\bodyX$ and $\bodyY$ of total
    complexity $n$, one can compute, in linear time,
    $\scaleSimX{\bodyX}{\bodyY} $, and the translation that realizes
    it.
\end{lemma}

\section{Building blocks}
\seclab{building}

\subsection{A better convex approximation in the plane}
\seclab{better:approx}

Let $\bbox$ be the minimum volume bounding box of some bounded convex
set $K \subseteq \Re^d$. We have that $v + c_d \bbox \subseteq K
\subseteq \bbox$ \cite{h-gaa-11}, for some vector $v$ and a constant
$c_d$ which depends only on the dimension $d$.  This approximation can
be computed in $O(n)$ time \cite{bh-eamvb-01}, where $n$ is the number
of vertices of the convex-hull of $K$.  The more powerful result
showing that a convex body can be approximated by an ellipsoid (up to
a scaling factor of $d$), is known as John's Theorem \cite{h-gaa-11}.

We need the following variant of the algorithm of Barequet and
Har-Peled \cite{bh-eamvb-01}.

\begin{lemma}
    \lemlab{fast:convex:approx}%
    Given a convex polygon $\bodyC$ in the plane, with $n$ vertices,
    one can compute, in linear time, a rectangle $\rectC$ and a point
    $\vecC$, such that $\vecC + \rectC \subseteq \bodyC \subseteq
    \vecC + 5 \rectC$.
\end{lemma}

\noindent\begin{minipage}{\linewidth}
\parpic[r]{%
   \begin{minipage}{0.3\linewidth}
       \includegraphics[width=0.99\linewidth]{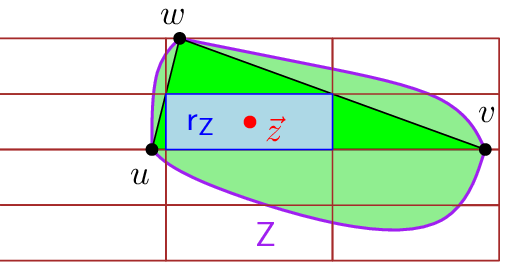}
   \end{minipage}}

\SaveIndent

\noindent\begin{minipage}{0.65\linewidth}
    \begin{proof}
        \RestoreIndent This is all well known, and we include the
        details for the sake of completeness. Using rotating caliper
        \cite{t-sgprc-83} compute the two vertices $u$ and $v$ of
        $\bodyC$ realizing its diameter.  Let $w$ be the vertex of
        $\bodyC$ furthest away from $u v$, Consider the rectangle
        $\rectC'$ having its base on $uv$, having half the height of
        $\triangle uv w$, and contained inside this triangle. Now, let
        $\vecC$ be the center of $\rectC'$, and set $\rectC = \rectC'
        - \vecC$, see figure on the right. It is now easy to verify
        that the claim holds with $\rectC$ and $\vecC$.  
        \InESAVer{ \qed}
    \end{proof}
\end{minipage}
\end{minipage}

\begin{observation}%
    \obslab{transformation}%
    Given two bodies $\bodyX, \bodyY \subseteq \Re^2$ and a
    non-singular affine transformation $\transform$, we have $\ds
    \frac{\areaX{\bodyX}}{ \areaX{\bodyY}}%
    =%
    \frac{\areaX{\transform(\bodyX)}}{ \areaX{\transform(\bodyY)}}$.
\end{observation}


Since a similar construction is described by Ahn \etal
\cite{acpsv-motpc-07}, we delegate the proof of this lemma to
\apndref{fast:convex:approx}.

\begin{lemma}[\approxPolyWidth]%
    \lemlab{approx:scheme}%
    Given a convex polygon $\Poly$, and a parameter $m > 0$, we can
    compute, in $O( \cardin{\Poly})$ time, a convex polygon $\Poly'$
    with $O(m)$ vertices, such that
    \begin{inparaenum}[(i)]
        \item $\Poly' \subseteq \Poly$, and
        \item for any point $\pnt \in \Poly$, its distance from
        $\Poly'$ is at most $\WidthX{\Poly} / m$, where
        $\WidthX{\Poly}$ is the width of $\Poly$.
    \end{inparaenum}
\end{lemma}


\subsection{The level set of the area of overlap %
   function}
\seclab{level:set}

\begin{definition}
    The superlevel set of a function $\funcA: \Re^d \rightarrow \Re$,
    for a value $\alpha$ is the set
    \begin{math}
        \LevelSetX{\alpha}{\funcA} = \brc{ \pnt \in \Re^d\sep{
              \funcA\pth{\pnt} \geq \alpha}}.
    \end{math}
    We will refer to it as the \emphi{$\alpha$-slice} of $\funcA$.
\end{definition}

\begin{lemma}%
    \lemlab{levelset}%
    Given two convex polygons $\bodyX$ and $\bodyY$, the slice
    $\bodyC = \LevelSetX{\alpha}{\AreaOverlap{\bodyX}{\trans +
          \bodyY}}$ is convex, and has complexity $O\pth{ m}$, where
    $m = \cardin{\bodyX} \cardin{\bodyY}$.  Furthermore, given a point
    $\pnt \in \bodyC$, the convex body $\bodyC$ can be computed in $O(
    m \log m )$ time.
\end{lemma}

\begin{proof:in:appendix}{\lemref{levelset}}
    Along any line, the function $\AreaOverlap{\bodyX}{\trans +
       \bodyY}$ is unimodal \cite{bcdkt-cmotc-98}. This directly
    implies that on any segment joining two points of the boundary of
    an $\alpha$-slice, the function will have values greater than
    $\alpha$, and hence, it will lie inside the $\alpha$-slice.
    Therefore, $\bodyC$ is convex.
    
    As one translates $\bodyY$ over $\bodyX$, as long as the same
    pairs of edges intersect, the function governing the overlap
    function will remain the same quadratic polynomial -- the function
    changes form whenever a vertex is being swept over.
    
    Now, consider an edge $e$ of $\bodyX$ and an edge $e'$ of
    $\bodyY$.  All the translations $\trans$ for which $e$ intersects
    $e'+\trans$ map out a parallelogram $\pi( e, e')$, which has edges
    parallel to $e$ and $e'$.  The edges of these parallelograms
    correspond to translations, where some vertex of $\bodyY$ lies on
    an edge of $\bodyX$, or vice versa, and the vertices correspond to
    those translations where some vertex of $\bodyX$ coincides with
    some vertex of $\bodyY$.
    
    So, consider the arrangement $\Arr$ of all lines passing
    through the edges of these parallelograms -- it is defined by
    $O\pth{\cardin{\bodyX} \cardin{\bodyY} }$ lines, and as such has
    complexity $O\pth{\cardin{\bodyX}^2 \cardin{\bodyY}^2}$
    overall. By convexity, the boundary $\gamma = \partial \bodyC$ 
    intersects every line at most twice, and hence the curve
    $\gamma$ visits at most $O\pth{\cardin{\bodyX} \cardin{\bodyY}}$
    faces of this arrangement (we count a visit to the same face with
    multiplicity). Thus, $\gamma$ can be broken into
    $O\pth{\cardin{\bodyX} \cardin{\bodyY}}$ arcs, each one of them of
    constant descriptive complexity (i.e., its the boundary of a slice
    of a quadratic function, between intersection points of this curve
    with two lines).
    
    As for the algorithm, set $\trans = \pnt$, and compute the overlap
    function value at $\pnt$ by sweeping the two polygons, computing
    their intersection polygon, and then computing the vertical
    decomposition of this intersection polygon. It is easy to verify,
    that the overlap function is a quadratic function of $\trans$, and
    its the sum of the areas of the vertical trapezoids, each one of
    them can be decomposed into two triangles, and the coordinates of
    the vertices of each triangle are affine function of $\trans$. As
    such, the area of each triangle is a quadratic function, and
    adding them up give us the formula for the area of the whole
    overlap polygon. Now, as $\trans$ moves, every time the vertical
    decomposition changes as a vertex of one polygon is swept over by
    the boundary of the other polygon, this corresponds to a local
    change in the vertical decomposition of the overlap polygon, and
    the overlap function can be changed in constant time. Thus,
    starting from $\pnt$, we move $\trans$ up, updating the overlap
    function as necessary till we reach a point, where the overlap
    function has value $\alpha$. At that point, we trace out the outer
    zone of $\partial \bodyC$ -- since this zone is defined by $m =
    O\pth{ \cardin{\bodyX}\cardin{\bodyY}}$ rays, this zone has
    complexity $O(m)$ (note, that the exact bound on the worst case
    complexity of the faces of the zone inside a convex region is
    still open, and it is between $O(m)$ and $O(m \alpha(m))$, where
    $\alpha(n)$ is the inverse Ackerman function). Specifically, we
    are walking on the arrangement of $\Arr$ starting on a point on a
    line that is on $\partial\bodyC$, go clockwise on the edges of the
    face till we encounter $\partial\bodyC$ again, move to the next
    face at this point, and repeat. To perform this, we need an
    efficient data-structure for doing a walk in a planar arrangement
    \cite{h-twpa-00} -- in this case since this is an arrangement of
    lines, we can use a data-structure for dynamic maintenance of
    convex-hull. In particular, Brodal and Jacob presented a
    data-structure \cite{bj-dpch-02} that supports this kind of
    operations in $O( \log m )$ amortized time per update. As such,
    computing $\partial \bodyC$ takes $O(m \log m)$ time overall.
\end{proof:in:appendix}

\subsection{The shape of the polygon realizing the %
   maximum area overlap}
\seclab{shape:overlap}

In the following, all the ellipses being considered are centered in
the origin. 

\begin{lemma}%
    \lemlab{ellipse:overlap}%
    Given two ellipses $\Ellipse_1$ and $\Ellipse_2$, the translation
    which maximizes their area of overlap is the one in which their
    centers are the same points.
\end{lemma}

\begin{proof}
    Translate $\Ellipse_1$ and $\Ellipse_2$ such that their centers
    are at the origin. Consider any unit vector $\vecA$, translate
    $\Ellipse_2$ along the direction of $\vecA$, and consider the
    behavior of the overlap function $f(x) = \AreaOverlap{\Ellipse_1}
    {\Ellipse_2 + x\vecA \MakeBig }$, where $x$ varies from $-\infty$
    to $+\infty$. The function $f$ is unimodal \cite{bcdkt-cmotc-98}.
    By symmetry, we have
    \begin{align*}
        f(x)%
        = %
        \AreaOverlap{ \Ellipse_1}{ \Ellipse_2 + x\vecA\MakeBig}%
        =%
        \AreaOverlap{ -\Ellipse_1}{\, -\pth{\Ellipse_2 +
              x\vecA}\MakeBig}%
        =%
        \AreaOverlap{\Ellipse_1}{\Ellipse_2 - x\vecA \MakeBig}%
        =%
        f(-x),
    \end{align*}
    as $\Ellipse_i = - \Ellipse_i$.  If the maximum is attained at
    $x\neq 0$, we will get another maximum at $-x$, which implies, as
    $f$ unimodal, that $f(0) = f(x) = f(-x)$, as desired.
        \InESAVer{ \qed}
\end{proof}

\begin{figure}
    \centerline{%
       \begin{minipage}[b]{0.48\linewidth}
           \centering
           \includegraphics[width=0.98\linewidth]{\si{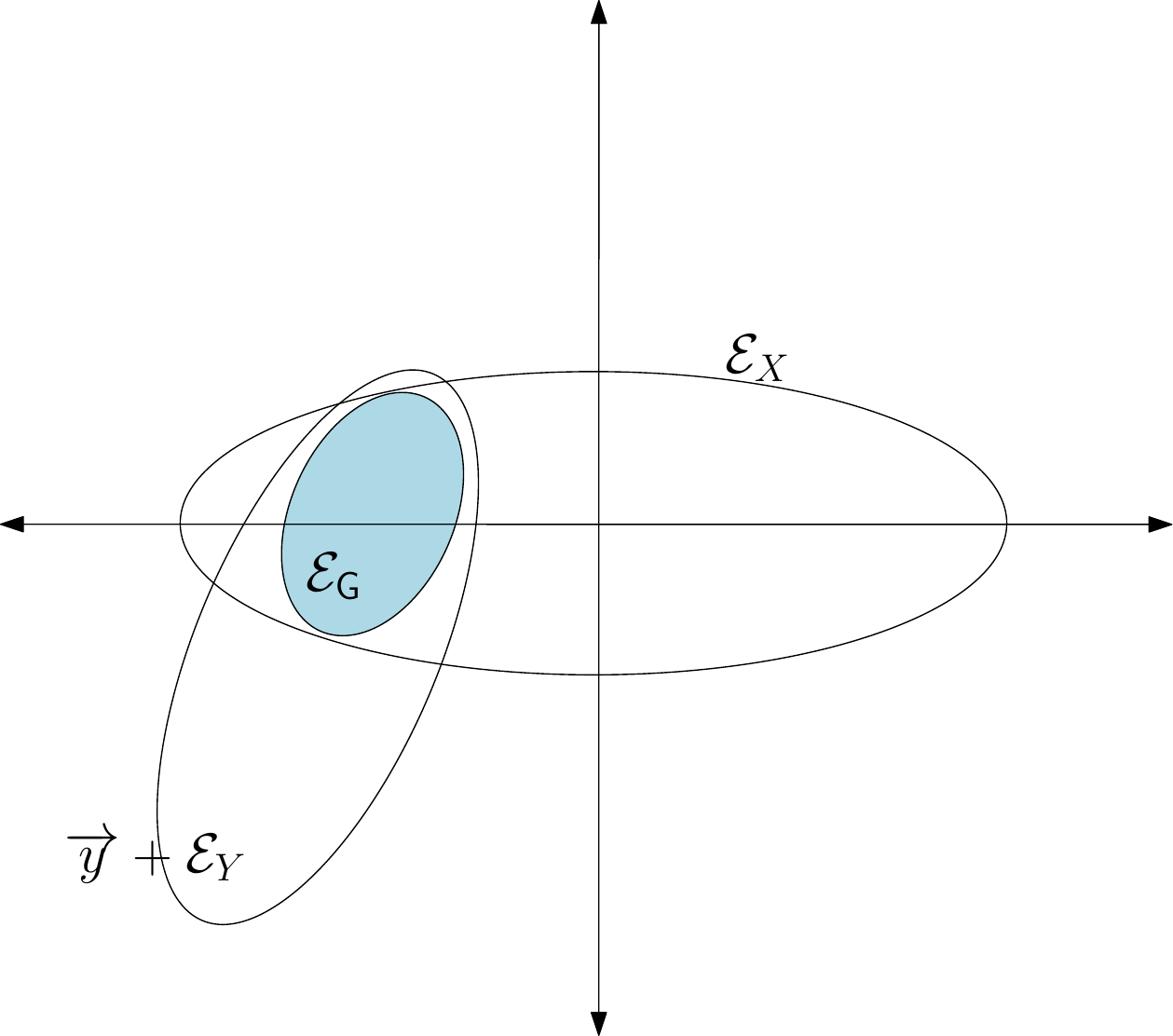}}
           \caption{}
           \figlab{ellipse:proof:1}
       \end{minipage}
       \begin{minipage}[b]{0.48\linewidth}
           \centering
           \includegraphics[width=0.98\linewidth]{\si{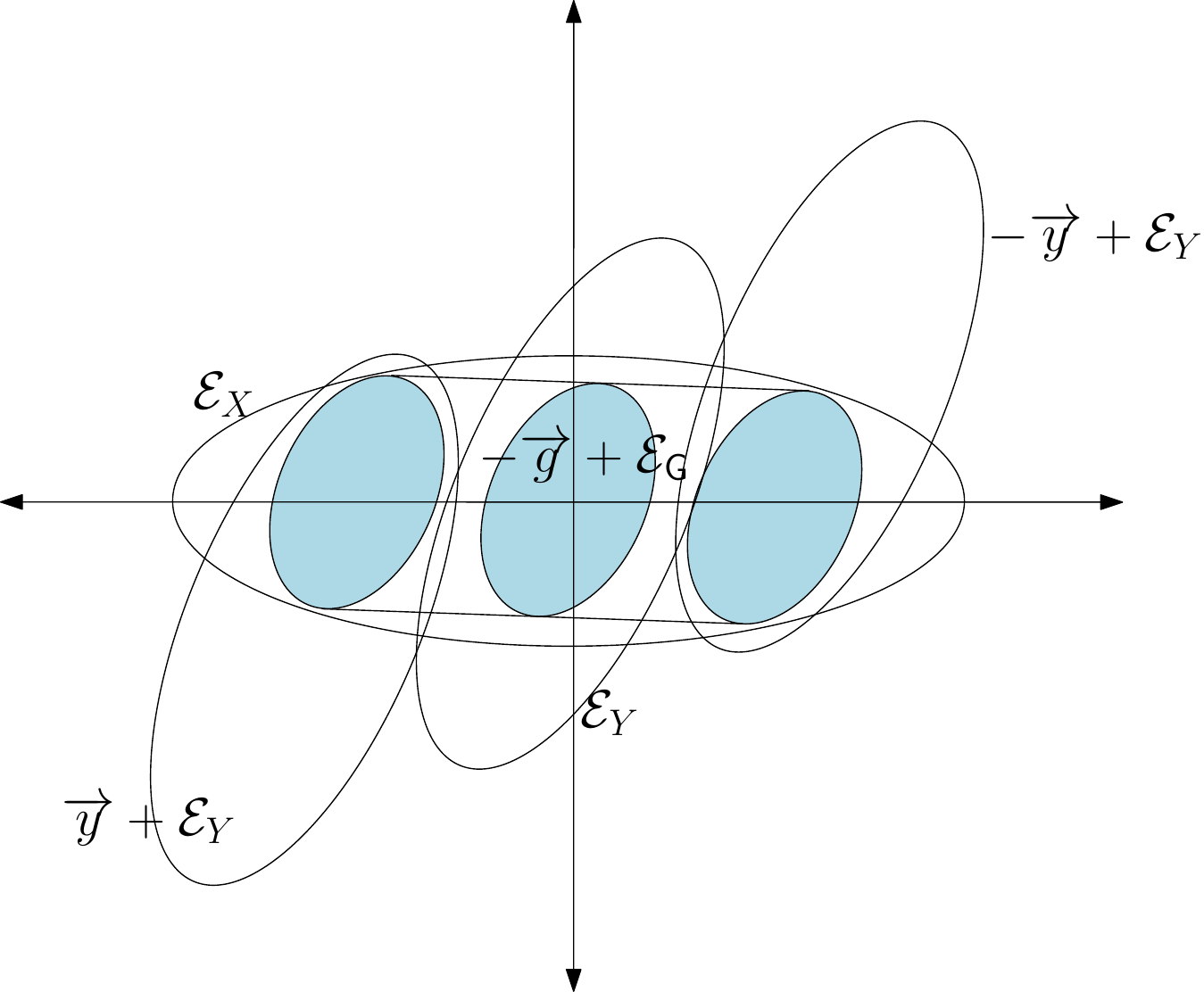}}
           \caption{}
           \figlab{ellipse:proof:2}
       \end{minipage}%
    }
\end{figure}

\begin{lemma}
    \lemlab{ellipse:overlap:center}%
    Consider two ellipses $\elpsX$ and $\elpsY$ in the plane, and
    consider any two vectors $\vecX$ and $\vecY$, then there is a
    vector $\vecA$ such that $\vecA + \pth[]{\vecX + \elpsX} \cap
    \pth[]{\vecY +\elpsY} \subseteq 2\elpsX \cap 2\elpsY$.
\end{lemma}
\begin{proof}
    For the sake of simplicity of exposition, assume that $\vecX =
    0$. Now, consider the intersection $\bodyG = \elpsX \cap
    \pth[]{\vecY + \elpsY}$, and let $\elpsG$ be the largest area
    ellipse contained inside $\bodyG$. John's theorem implies that
    there is a translation vector $\vecG$, such that $\vecG + \elpsG
    \subseteq \bodyG \subseteq \vecG + 2\elpsG$, see
    \figref{ellipse:proof:1}.
    
    Observe that $ \vecG + \elpsG \subseteq \elpsX$, and by the
    symmetry of $\elpsG$ and $\elpsX$, we have that $-\vecG + \elpsG =
    -\vecG - \elpsG \subseteq -\elpsX = \elpsX$. This by convexity
    implies that $\elpsG \subseteq \elpsX$. A similar argument implies
    that $\elpsG \subseteq \elpsY$. As such, $\elpsG \subseteq \elpsX
    \cap \elpsY$.
    
    Thus, we have that $\bodyG \subseteq \vecG + 2\elpsG \subseteq
    \vecG + 2\elpsX \cap 2\elpsY$, as desired.
        \InESAVer{ \qed}
\end{proof}

\begin{lemma}[{{\cite[Lemma 22.5]{h-gaa-11}}}]%
    \lemlab{big:inner:ball}%
    Any convex set $K \subseteq \Re^d$ contained in a unit square,
    contains a ball of radius $\areaX{K} / 8$
\end{lemma}

The following lemma is one of our key insights -- the maximum area of
intersection of two polygons contains any intersection of translated
copies of these polygons up to translation and a constant factor
scaling.

\begin{lemma}%
    \lemlab{inclusion}%
    Let $\bodyX$ and $\bodyY$ be two convex polygons, and let
    $\bodyMax$ be the polygon realizing their maximum area of
    intersection under translation. Let $\vecA$ be any vector in the
    plane, and consider the polygon $\bodyD = \bodyX \cap \pth[]{\vecA
       + \bodyY}$, then there exists a vector $\vecB$ such that,
    $\vecB + \bodyD \subseteq c_0\bodyMax$, for some fixed constant
    $c_0$.
\end{lemma}

\begin{proof}
    Let $\elpsX$ (resp., $\elpsY$) denote the maximum area ellipse
    (centered at the origin) contained inside $\bodyX$
    (resp. $\bodyY$). By John's Theorem, we have $\vecX + \elpsX
    \subseteq \bodyX \subseteq \vecX + 2 \elpsX$ and $\vecY + \elpsY
    \subseteq \bodyY \subseteq \vecY + 2 \elpsY$, where $\vecX, \vecY$
    are some vector.  Let $\bodyB = \elpsX \cap \elpsY$, and let
    $\elpsB$ be the maximum area ellipse contained inside $\bodyB$.
    Observe that $\bodyB$ is symmetric and centered at the origin, and
    by John's theorem $\elpsB \subseteq \bodyB \subseteq2\elpsB$.
    
    By \lemref{ellipse:overlap:center}, there are vectors $\vecC$ and
    $\vecD$, such that
    \begin{align*}
        \bodyD%
        &=%
        \bodyX \cap \pth[]{ \vecA + \bodyY}%
        \subseteq%
        \pth{ \vecX + 2\elpsX} \cap \pth[]{\vecC + \vecY + 2\elpsY}%
        \subseteq%
        \vecD + 4\elpsX \cap 4\elpsY%
        =%
        \vecD + 4\bodyB\\%
        &\subseteq %
        \vecD + 8\elpsB.
    \end{align*}
    Applying a similar argument, we have that $\bodyMax \subseteq
    \vecMax + 8\elpsB$, for some vector $\vecMax$.
    
    Apply the linear transformation that maps $\elpsB$ to
    $\diskX{1/16}$, where $\diskX{r}$ denotes the disk of radius $r$
    centered at the origin. By \obsref{transformation}, we can
    continue our discussion in the transformed coordinates. This
    implies that $\bodyMax - \vecMax \subseteq \diskX{1/2}$ (which is
    contained inside a unit square). By \lemref{big:inner:ball}, there
    is a vector $\vecXA$, such that $\vecXA +
    \diskX{\areaX{\bodyMax}/8} \subseteq \bodyMax$.
    
    Observe that $\bodyB = \elpsX \cap \elpsY \subseteq \pth[]{- \vecX
       + \bodyX} \cap \pth[]{- \vecY + \bodyY}$. As such, the area of
    $\bodyB$ must be smaller than the area of $\bodyMax$ (by the
    definition of $\bodyMax$).  We thus have
    \begin{math}
        \areaX{\bodyMax}%
        \geq%
        \areaX{\bodyB}%
        \geq%
        \areaX{\elpsB}%
        =%
        \areaX{\diskX{1/16}}%
    \end{math}
    which is a constant bounded away from zero.
    Therefore, 
    \begin{align*}
	\bodyD%
        &\subseteq%
        \vecD + 8\elpsB%
        =%
        \vecD + \diskX{\frac{1}{2}}%
        =%
        \vecD + \frac{4}{\areaX{\bodyMax}} \cdot
        \diskX{\frac{\areaX{\bodyMax}}{8} }%
        \\
        &\subseteq%
        \vecD + \frac{4}{\areaX{\bodyMax}} \pth{ \bodyMax - \vecXA},%
    \end{align*}
    which implies the claim.
        \InESAVer{ \qed}
\end{proof}

\subsubsection{Constant approximation to the maximum overlap}%

\begin{lemma}[\constApproxByRect]%
    \lemlab{constant:approx}%
    Let $\bodyX$ and $\bodyY$ be two convex polygons, and let
    $\bodyMax$ be the polygon realizing their maximum area
    intersection under translation. Then, one can compute, in $O\pth{
       \cardin{\bodyX} + \cardin{\bodyY}}$ time, a rectangle $\rectA$,
    such that $\rectA \subseteq \vecA + \bodyMax \subseteq \constRect
    \rectA$, where $\constRect$ is a constant. That is, one can
    compute a constant factor approximation to the maximum area
    overlap in linear time.
    
    Furthermore, for any translation $\transY$, we have that $\bodyX
    \cap \pth[]{\bodyY + \transY } \sqsubseteq \constRect \rectA $.
\end{lemma}
\begin{proof}
    We are going to implement the algorithmic proof of
    \lemref{inclusion}.  Instead of John's ellipsoid we use the
    rectangle of \lemref{fast:convex:approx}. Clearly, the proof of
    \lemref{inclusion} goes through with the constants being somewhat
    worse. Specifically, we compute, in linear time, vectors $\vecX,
    \vecY$, and rectangles $\rectX, \rectY$, such that $\vecX + \rectX
    \subseteq \bodyX \subseteq \vecX + 5 \rectX$ and $\vecY + \rectY
    \subseteq \bodyY \subseteq \vecY + 5 \rectY$. Again, compute a
    rectangle $\rectM$, such that $\rectM/5 \subseteq \rectX \cap
    \rectY \subseteq \rectM$. Arguing as in \lemref{inclusion}, and
    setting $\rectA = \rectM / \constC$, for some constant $\constC$,
    is the desired rectangle.
    \InESAVer{ \qed}
\end{proof}


\section{Approximating the overlap %
   function of convex polygons}
\seclab{approx:overlap:convex}

\begin{defn}
    Given two convex polygons $\bodyX$ and $\bodyY$ in the plane, of
    total complexity $n$, and parameters $\eps \in (0,1)$, $
    \numPolygons$, $\totalN$, a function $\AOA{\trans}$ is
    \emphi{$(\eps, \numPolygons, \totalN)$-approximation} of
    $\AreaOverlap{\bodyX}{ \trans + \bodyY }$, if the following
    conditions hold:
    \begin{compactenum}[\qquad(A)]
        \item $\forall \trans \in \Re^2$, we have
        \begin{math}
            \cardin{\AreaOverlap{\bodyX}{ \trans + \bodyY } -
               \AOA{\trans }} \leq \eps \maxOverlap{\bodyX}{\bodyY}.
        \end{math}
        
        \item There are convex polygons $\Poly_1, \ldots,
        \Poly_\numPolygons$, each of maximum complexity $\totalN$, such
        that inside every face of the arrangement $\ArrC = \ArrX{
        \Poly_1, \ldots, \Poly_\numPolygons}$, the approximation
        function $\AOA{\trans}$ is the same quadratic function.
    \end{compactenum}
    That is, the total descriptive complexity of $\AOA{\cdot}$ is the
    complexity of the arrangement $\ArrC$.
\end{defn}

\begin{algorithm}%
    \alglab{preprocessing}%
    The input is two convex polygons $\bodyX$ and $\bodyY$ in the
    plane, of total complexity $n$, and a parameter $\eps \in (0,1)$.
    As a first step, the algorithm is going to approximate $\bodyX$
    and $\bodyY$ as follows:
    \begin{compactenum}[\qquad(A)]
        \item $\rectM \leftarrow \constApproxByRect\pth{\bodyX,
           \bodyY}$, see \lemref{constant:approx}.

        \item $\Trans \leftarrow $ affine transformation that maps $2
        \constRect \rectM$ to $[0,1]^2$.

        \item $\bodyXTA \leftarrow
        \approxPolyWidth\pth{\Trans\pth{\bodyX}, \, N }$ and $\bodyYTA
        \leftarrow {\approxPolyWidth\pth{\Trans\pth{\bodyY}, \, N }}$.

        See \lemref{approx:scheme}, here $N = \ceil{ \constD/ \eps}$,
        and $\constD$ is a sufficiently large constant.

        \item $\bodyXA \leftarrow \Trans^{-1}\pth{ \bodyXTA }$ and
        $\bodyYA \leftarrow \Trans^{-1}\pth{ \bodyYTA }$.
    \end{compactenum}
\end{algorithm}

\subsection{If one polygon is smaller than the other}
\seclab{small}

\parpic[r]{\includegraphics{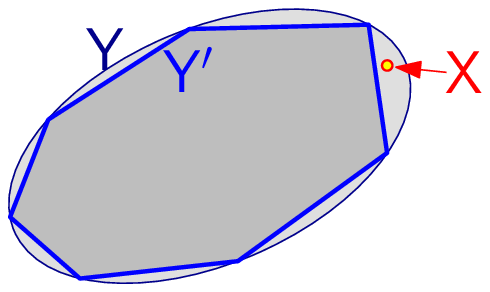}}

Assume, without loss of generality, that $\bodyX$ is smaller than
$\bodyY$, that is, $\bodyX$ can be translated so that it is entirely
contained inside $\bodyY$ (i.e., $\scaleSimX{\bodyX}{\bodyY} \leq 1$,
see \lemref{scaling}). The maximum area of overlap is now equal to
$\areaX{\bodyX}$.  The challenge is, that for any approximation of
$\bodyY$, we can always have a sufficiently small $\bodyX$ which can
be placed in $\bodyY \setminus \bodyYA$, as shown in the figure on the
right. Therefore for all those translations for which $\bodyX$ is
placed inside $\bodyY \setminus \bodyYA$, our approximation will show
zero overlap, even though the actual overlap is $\areaX{\bodyX}$.

To get around this problem, we will first approximate the smaller
polygon $\bodyX$, using our approximation scheme, to get polygon
$\bodyXA$, then we will compute level sets of the overlap function and
use them to approximate it.

\begin{lemma}
    \lemlab{bounded:width:2}%
    Given convex polygons $\bodyX$ and $\bodyY$, such that
    $\scaleSimX{\bodyX}{\bodyY} < 1$, and parameter $\eps > 0$, and
    let $\bodyXA$ be the approximation to $\bodyX$, as computed by
    \algref{preprocessing}.  Then, we have, for all translations
    $\trans \in \Re^2$, that
    \begin{math}
        \cardin{\MakeBig \AreaOverlap{\bodyXA}{\trans + \bodyY} -
           \AreaOverlap{\bodyX}{\trans + \bodyY}}%
        \leq%
        \eps \maxOverlap{\bodyX}{\bodyY}.
    \end{math}
\end{lemma}

\begin{proof}
    Consider the overlap of $\bodyXT = \Trans\pth{\bodyX}$ and
    $\bodyYT = \Trans\pth{\bodyY}$.  \lemref{constant:approx} implies
    that any intersection polygon of $\bodyXT$ and $\bodyYT$ can be
    contained (via translation) in $\Trans\pth{\constRect\rectM}$
    (which is a translation of the square $[0,1/2]^2$).  Clearly, in
    this case, $\bodyXT$ and $\bodyXTA$ can both be translated to be
    contained in this square, both contain a disk of constant radius,
    the maximum distance between $\bodyXT$ and $\bodyXTA$ is
    $O(\eps)$, and the total area of ${\bodyXT \setminus \bodyXTA}$ is
    $O(\eps)$, as the perimeter of $\bodyXT \leq 4$. Thus, setting
    $\constD$ to be sufficiently large, implies that $\areaX{\bodyXT
       \setminus \bodyXTA} \leq \eps \maxOverlap{\bodyXT}{\bodyYT}$,
    as $\maxOverlap{\bodyXT}{\bodyYT} = \Omega(1)$. This implies that
    \begin{math}
        \cardin{\MakeBig \AreaOverlap{\bodyXTA}{\trans + \bodyYT} -
           \AreaOverlap{\bodyXT}{\trans + \bodyYT}} \leq \eps
        \maxOverlap{\bodyXT}{\bodyYT},
    \end{math}
    which implies the claims by applying $\Trans^{-1}$ to both sides.
    \InESAVer{ \qed}
\end{proof}

Therefore, $\AreaOverlap{\bodyXA}{\trans + \bodyY}$ is a good
approximation for $\AreaOverlap{\bodyX}{\trans + \bodyY}$. However,
$\AreaOverlap{\bodyXA}{\trans + \bodyY}$ has complexity $O\pth{
   \cardin{\bodyXA}^2 \cardin{\bodyY}^2 }$ \cite{bcdkt-cmotc-98}, in
the worst case, which is still too high.

\begin{lemma}[\approxLevelSet]%
    \lemlab{approx:small:large}%
    Given two convex polygons $\bodyX$ and $\bodyY$, of total
    complexity $n$, and a parameter $\eps$, such that
    $\scaleSimX{\bodyX}{\bodyY} < 1$, then one can construct in $\ds
    O\pth{ n/\eps^2 }$ time, a $\pth[]{\eps, O(1/\eps^2), O(n/\eps^2)
    }$-approximation $\AOA{\cdot}$ to $\AreaOverlap{\bodyX}{\trans +
       \bodyY}$.
\end{lemma}

\begin{proof}
    There is a translation of $\bodyX$ such that it is contained
    completely in $\bodyY$. Approximate $\bodyX$ from the outside by a
    rectangle $\rectA$, using \lemref{fast:convex:approx}.  Next,
    spread a grid in $\rectA$ by partitioning each of its edges into
    $O(1/\eps)$ equal length intervals. Let $\PntSet$ be the set of
    points of the grid that are in $\bodyX$. It is easy to verify,
    that for any convex body $\bodyC$ and a translation $\trans$, we
    have
    \begin{align*}
        \cardin{\AreaOverlap{\bodyX}{ \trans + \bodyC } -
           \frac{\cardin{(\trans+\bodyC) \cap
                 \PntSet}}{\cardin{\PntSet}}}%
        \leq%
        \eps \, \areaX{\bodyX}
    \end{align*}
    Namely, to approximate the overlap area for $\trans + \bodyY$, we
    need to count the number of points of $\PntSet$ that it covers. To
    this end, for each point $\pnt \in \PntSet$, we generate a
    $180^{\degree}$ rotated and translated copy of $\bodyY$, denoted
    by $\bodyY_\pnt'$, such that $\pnt \in \trans + \bodyY$ if and
    only if $\trans \in \bodyY_\pnt'$.

    Clearly, the generated set of polygons is the desired 
    $\pth[]{\eps, O(1/\eps^2), O(n/\eps^2) }$-approximation $\AOA{\cdot}$
    to $\AreaOverlap{\bodyX}{\trans + \bodyY}$.

    The time to build this approximation is $O(n/\eps^2)$.
\end{proof}

We next describe a slightly slower algorithm that generates a slightly
better approximation.

\begin{lemma}[\approxLevelSet]%
    \lemlab{approx:small:large:better}%
    Given two convex polygons $\bodyX$ and $\bodyY$, of total
    complexity $n$, and a parameter $\eps$, such that
    $\scaleSimX{\bodyX}{\bodyY} < 1$, then one can construct in $\ds
    O\pth{ \eps^{-2} n \log n}$ time, a $\pth[]{\eps, O(1/\eps),
       O(n/\eps) }$-approximation $\AOA{\cdot}$ to
    $\AreaOverlap{\bodyX}{\trans + \bodyY}$.
\end{lemma}

\begin{proof:in:appendix}{\lemref{approx:small:large:better}}
    We compute $\bodyXA$, as above with approximation parameter $\epsA
    = \eps/4$.  Next, using the algorithm of \cite{bcdkt-cmotc-98}, we
    compute $\maxOverlap{\bodyXA}{ \bodyY }$ in $O( n' \log n') = O(n
    \log n)$ time, where $n' = \cardin{\bodyXA} + \cardin{\bodyY} =
    O(n)$.

    We approximate the function $\AreaOverlap{\bodyXA}{ \bodyY+\trans
    }$ by constructing the $\alpha_i$-slices, where $\alpha_i = \min(
    1, (i+1)\epsA) \maxOverlap{\bodyXA}{ \bodyY }$, for $i = 0, 1,
    \ldots, M$, where $M = \ceil{ 1/\epsA}$.  To this end, we deploy
    the algorithm of \lemref{levelset} for each slice, which takes $O(
    m \log m)$ time, where $m = \cardin{\bodyXA}{\cardin{\bodyY}} =
    O(n/\eps)$.

    Let us denote the $i$\th region constructed, by $\slice_i$, that
    is, $\slice_i$ is a convex figure whose boundary corresponds to
    all translations $\trans$ such that
    $\AreaOverlap{\bodyXA}{\bodyY+\trans}= \alpha_i$.  Clearly,
    $\slice_{i+1}$ lies entirely within $\slice_i$. Given the
    description of $\slice_i$'s, for any translation $\trans$,  we 
    define
    \begin{align*}
        \AOA{\trans}%
        =%
        \begin{cases}
            \alpha_i & \trans \in \slice_i \setminus \slice_{i+1} \\
            0 & \trans \notin \slice_0.
        \end{cases}
    \end{align*}
       
    It is now straightforward to verify this is the desired
    approximation. Indeed, for $\trans \notin \slice_0$, we have by
    \lemref{bounded:width:2}, that
    \begin{align*}
        \cardin{\MakeBig \AOA{\trans} - \AreaOverlap{\bodyX}{\trans +
              \bodyY}}%
        &\leq%
        \cardin{\MakeBig \AOA{\trans} - \AreaOverlap{\bodyXA}{\trans +
              \bodyY}}%
        +%
        \cardin{\MakeBig \AreaOverlap{\bodyXA}{\trans + \bodyY} -
           \AreaOverlap{\bodyX}{\trans + \bodyY}}%
        \InESAVer{\\&}%
        \leq%
        \alpha_0 + \epsA \maxOverlap{\bodyX}{\bodyY}%
        \InNotESAVer{\\&}%
        =%
        \epsA \pth{ \maxOverlap{\bodyXA}{ \bodyY } +
           \maxOverlap{\bodyX}{\bodyY}}%
        \InESAVer{\\&}%
        \leq%
        2\epsA \maxOverlap{\bodyX}{\bodyY}%
        =%
        (\eps/2) \maxOverlap{\bodyX}{\bodyY},
    \end{align*}
    as desired. Similarly, if $\trans \in \slice_i \setminus
    \slice_{i+1}$ then $\alpha_i \leq \AreaOverlap{\bodyXA}{\trans +
       \bodyY} \leq \alpha_{i+1}$, and
    \begin{align*}
        \cardin{\MakeBig \ApproxOverlapC(\trans) -
           \AreaOverlap{\bodyX}{\trans + \bodyY}}%
        &\leq%
        \cardin{\MakeBig \ApproxOverlapC(\trans) -
           \AreaOverlap{\bodyXA}{\trans + \bodyY}}%
        +%
        \cardin{\MakeBig \AreaOverlap{\bodyXA}{\trans + \bodyY} -
           \AreaOverlap{\bodyX}{\trans + \bodyY}}%
        \\
        &\leq%
        \alpha_{i+1} - \alpha_i + \epsA \maxOverlap{\bodyX}{\bodyY}%
        =%
        \epsA \maxOverlap{\bodyXA}{ \bodyY } + \epsA
        \maxOverlap{\bodyX}{\bodyY}%
        \\
        & \leq%
        (\eps/2) \maxOverlap{\bodyX}{\bodyY}.
    \end{align*}
    Clearly, this function is defined by an onion like set of $
    \numPolygons = O(1/\eps)$ polygons, and the maximum complexity of
    these polygons is $\totalN = O(m ) = O(n/\eps)$.

    The overall running time is dominated by computing the slices,
    which takes overall $O( \numPolygons m \log m ) \linebreak[0] = O(
    n \eps^{-2} \log n \eps^{-1})$ time.  \InESAVer{ \qed}
\end{proof:in:appendix}

\subsection{If the two polygons are incomparable}
\seclab{large}

The more interesting case, is when the maximum intersection of
$\bodyX$ and $\bodyY$ is significantly smaller than both polygons;
that is, $\scaleSimX{\bodyX}{\bodyY} \geq 1$ and
$\scaleSimX{\bodyY}{\bodyX} \geq 1$. Surprisingly, in this case, we
can approximate both polygons simultaneously.

\begin{lemma}%
    \lemlab{bounded:width}%
    Given convex polygons $\bodyX$ and $\bodyY$, such that
    $\scaleSimX{\bodyX}{\bodyY} \geq 1$ and
    $\scaleSimX{\bodyY}{\bodyX} \geq 1$, then the widths of $\bodyXT =
    \Trans\pth{\bodyX}$ and $\bodyYT = \Trans\pth{\bodyY}$, as
    computed by \algref{preprocessing}, are bounded by $7$.
\end{lemma}

\begin{proof:in:appendix}{\lemref{bounded:width}}
    Let $\wdX = \widthX{\bodyXT}$ and $\wdY =
    \widthX{\Trans\pth{\bodyY}}$.  By \lemref{width:to:ball}, we have
    that $\diskX{\wdX/3.5 } \sqsubseteq \bodyXT$ and $\diskX{\wdY/3.5
    } \sqsubseteq \bodyYT$.
    
    So, assume for the sake of contradiction, that $\wdX \geq 7$. This
    implies that $\bodyXT$ contains a disk of radius $2$, which in
    turn contains the unit square. In particular, let $\ell =
    \diameterX{\bodyYT}$, if $\ell < 2$, then $\bodyYT$ is contained
    in a disk of radius $2$, implies that $\bodyYT \sqsubseteq
    \bodyXT$, a contradiction.
    
    Otherwise, if $\ell > 2$ then there is a translation of $\bodyYT$
    such that its intersection with $\bodyYT$ has length $>2$ (indeed,
    consider the segment realizing the diameter of $\bodyYT$, and
    translate it so its middle point is in the center of the disk of
    radius $2$ inside $\bodyXT$). But then, this intersection is not
    contained in $[0,1]^2$ under any translation, which contradicts
    \lemref{constant:approx}.
    
    The case $\wdY \geq 7$ is handled in a similar fashion.
\end{proof:in:appendix}

\begin{lemma}
    \lemlab{approx:incomparable}%
    Given two convex polygons $\bodyX$ and $\bodyY$, of total
    complexity $n$, and a parameter $\eps$, such that
    $\scaleSimX{\bodyX}{\bodyY} \geq 1$ and
    $\scaleSimX{\bodyY}{\bodyX} \geq 1$, then one can construct in
    $\ds O\pth{ n + 1/\eps^2 }$ time, a $\pth[]{\eps, O(1/\eps),
       O(1/\eps) }$-approximation $\AOA{\cdot}$ to
    $\AreaOverlap{\bodyX}{\trans + \bodyY}$.
\end{lemma}

\begin{proof:in:appendix}{\lemref{approx:incomparable}}
    Use \algref{preprocessing} to compute $\bodyXA$ and $\bodyYA$,
    both of complexity $O(1/\eps)$.  We set $\AOA{\trans} =
    \AreaOverlap{\bodyXA}{\trans + \bodyYA}$.  \cite{bcdkt-cmotc-98}
    describes how to describe the function
    $\AreaOverlap{\bodyXA}{\trans + \bodyYA}$ as an arrangement of
    $O\pth{\cardin{\bodyXA} + \cardin{\bodyYA}}$ polygons, each of
    complexity $\cardin{\bodyXA}$ or $\cardin{\bodyYA}$. Thus, this
    result in the desired approximation, that has total complexity
    $O(1/\eps^2)$, and that can be computed in $O(n + 1/\eps^2)$ time.

    In the following, we use the notation of \algref{preprocessing}:
    $\bodyXT = \Trans\pth{\bodyX}$ and $\bodyYT = \Trans\pth{\bodyY}$,
    $\bodyXTA = \approxPolyWidth\pth{\Trans\pth{\bodyX}, \, N }$, and
    $\bodyYTA = \approxPolyWidth\pth{\Trans\pth{\bodyX}, \, N }$.

    \lemref{constant:approx} implies that any intersection polygon of
    $\bodyXT$ and $\bodyYT$ can be contained in
    $\Trans\pth{\constRect\rectM}$.  The error due to approximation of
    $\bodyXT$, is $\areaX{\bodyXT \setminus \bodyXTA}$. The part of
    this error that can contribute to the area of overlap, is bounded
    by portion of $\bodyXT \setminus \bodyXTA$, which can be included
    inside $\Trans\pth{\constRect\rectM}$.
    
    \begin{figure}
        \centering
        \includegraphics[height=4cm]{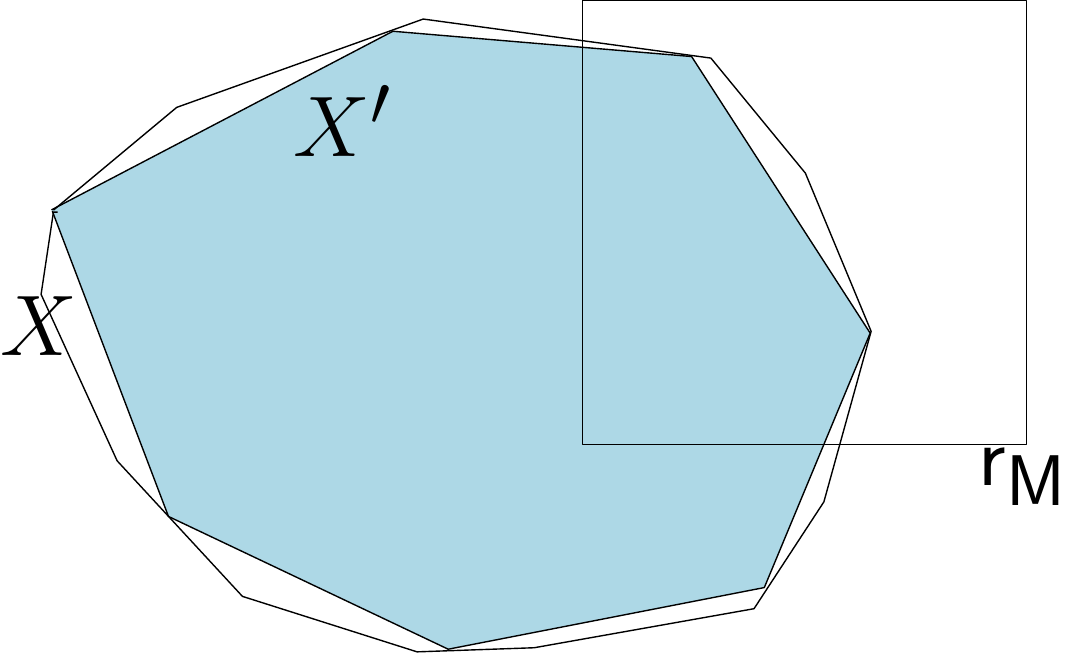}
        \caption{Error bound.}
        \figlab{errorbound}
    \end{figure}

    The length of the boundary of $\bodyXTA$, which can be placed
    inside $\Trans\pth{\constRect\rectM}$, is bounded by the perimeter
    of $\Trans\pth{\constRect\rectM}$. Also, our approximation scheme
    (\lemref{approx:scheme}) ensures that the distance between
    $\bodyXT$ and $\bodyXTA$ (along the direction of shortest
    diameter) is bounded by $\WidthX{ \bodyXT } /N $, where $N$ is a
    chosen parameter.  Consider any fixed $\trans$ -- for the sake of
    simplicity of exposition, we assume the intersection $\bodyXT \cap
    \pth[]{\trans+\bodyYT}$ as being inside the square $\Trans\pth{
       \constRect\rectM}$ (if not, we can translate $\Trans\pth{
       \constRect\rectM}$ so this assumption holds).  As $N = \ceil{
       \constD/ \eps}$, and for $\constD$ sufficiently large constant,
    the error due to the approximation of $\bodyXT$ is bounded by
    \begin{align*}
        \errBX%
        &=%
        \cardin{ \MakeBig \AreaOverlap{\bodyXTA}{\trans + \bodyYT} -
           \AreaOverlap{\bodyXT}{\trans + \bodyYT} }%
        =%
        \cardin{ \MakeBig \AreaOverlap{\bodyXT \setminus
              \bodyXTA}{\trans + \bodyYT} }%
        \\
        &\leq%
        \mathrm{perimeter}\pth{ \MakeBig \bodyXT \cap \pth[]{\trans +
              \bodyYT}} \times \frac{\WidthX{\bodyXT}}{N} \leq%
        \mathrm{Perimeter}\pth{\Trans\pth{\constRect\rectM}} \times
        \frac{\WidthX{\bodyXT}}{N}%
        \\%
        &\leq%
        4 \times \frac{7}{N} \leq%
        \frac{\eps}{4} \maxOverlap{\bodyXT}{\bodyYT},
    \end{align*}
    by \lemref{bounded:width}, and since by \lemref{constant:approx}
    we have $\maxOverlap{\bodyXT}{\bodyYT} = \Omega(1)$. A symmetric
    argument works for the error $\errBY$ of the overlap caused by the
    approximation of $\bodyYT$ by $\bodyYTA$. We conclude that 
    \begin{align*}
        \err_\Trans%
        &=%
        \cardin{ \MakeBig \AreaOverlap{\bodyXTA}{\trans + \bodyYTA} -
           \AreaOverlap{\bodyXT}{\trans + \bodyYT} }%
        \\%
        &\leq%
        \cardin{ \MakeBig%
           \AreaOverlap{\bodyXTA}{\trans + \bodyYTA} -
           \AreaOverlap{\bodyXT}{\trans + \bodyYTA}}%
        +%
        \cardin{ \MakeBig%
           \AreaOverlap{\bodyXT}{\trans + \bodyYTA} -
           \AreaOverlap{\bodyXT}{\trans + \bodyYT} }%
        \\
        &\leq%
        \cardin{ \MakeBig%
           \AreaOverlap{\bodyXTA}{\trans + \bodyYT} -
           \AreaOverlap{\bodyXT}{\trans + \bodyYT}}%
        +%
        \errBY%
        \leq %
        \errBX + \errBY%
        \leq%
        \frac{\eps}{2} \maxOverlap{\bodyXT}{\bodyYT},
    \end{align*}
    since $\bodyYTA \subseteq \bodyYT$. By applying $\Trans^{-1}$ to
    the above, we get that for any translation $\trans$, it holds
    \begin{align*}
        \err%
        &=%
        \cardin{ \MakeBig \AOA{\trans} -\AreaOverlap{\bodyX}{\trans +
              \bodyY} }%
        = \cardin{ \MakeBig \AreaOverlap{\bodyXA}{\trans + \bodyYA} -
           \AreaOverlap{\bodyX}{\trans + \bodyY} }%
        \leq%
        \frac{\eps}{2} \maxOverlap{\bodyX}{\bodyY}. %
    \end{align*}
    \InESAVer{ \qed}
\end{proof:in:appendix}

\mysubsubsection{The result}

By combining \lemref{approx:small:large} and
\lemref{approx:incomparable} (deciding which one to apply can be done
by computing $\scaleSimX{\bodyX}{\bodyY}$ and
$\scaleSimX{\bodyY}{\bodyX}$, which takes $O(n)$ time), we get the
following.

\begin{lemma}
    \lemlab{approx:convex:polygons}%
    Given two convex polygons $\bodyX$ and $\bodyY$, of total
    complexity $n$, and a parameter $\eps$, one can construct in $\ds
    O\pth{ n /\eps^2 }$ time, a $\pth[]{\eps, O(1/\eps^2), O(n/\eps^2)
    }$-approximation $\AOA{\cdot}$ to $\AreaOverlap{\bodyX}{\trans +
       \bodyY}$.
\end{lemma}


\section{Approximating the Maximum Overlap of %
   Polygons}
\seclab{approx:main}

The input is two polygons $\Poly$ and $\PolyA$ in the plane, of total
complexity $n$, each of them can be decomposed into at most $k$ convex
polygons. Our purpose is to find the translation that maximizes the
area of overlap.


\paragraph{The Algorithm.}%
We decompose the polygons $\Poly$ and $\PolyA$ into minimum number of
interior disjoint convex polygons \cite{ks-otbcd-02}, in time $O\pth{
   n+ k^2 \min(k^2, n) }$ (some of these convex polygons can be
empty). Then, for every pair $\Poly_i$, $\PolyA_j$, we compute an
$\pth[]{\epsA, O(1/\epsA^2), O(n/\epsA^2) }$-approximation
$\aFunc_{ij}$ to the overlap function of $\Poly_i$ and $\PolyA_j$,
using \lemref{approx:convex:polygons}, where $\epsA = \eps/k^2$.

Next, as each function $\aFunc_{ij}$ is defined by an arrangement
defined by $O(1/\epsA^2)$ polygons, we overlay all these arrangements
together, and compute for each face of the arrangement the function
$\aFunc = \sum_{i,j} \aFunc_{ij}$. Inside such a face this function is
the same, and it is a quadratic function. We then find the global
maximum of this function, and return it as the desired approximation.


\paragraph{Analysis -- Quality of approximation.}

For any translation $\trans$, we have that 
\begin{align*}
    \cardin{\MakeBig \AreaOverlap{\Poly}{\trans + \PolyA} -
       \aFunc(\trans)}%
    &\leq%
    \sum_{i=1}^k \sum_{j=1}^k \cardin{\MakeBig
       \AreaOverlap{\Poly_i}{\trans + \PolyA_j} -
       \aFunc_{ij}(\trans)}%
    \leq%
    \sum_{i=1}^k \sum_{j=1}^k \epsA \maxOverlap{\Poly_i}{\PolyA_j}%
    \InESAVer{\\&}%
    \leq%
    \epsA k^2 \maxOverlap{\Poly}{\PolyA}%
    \InNotESAVer{\\&}%
    \leq%
    \eps \maxOverlap{\Poly}{\PolyA}.
\end{align*}

\paragraph{Analysis -- Running time.}
Computing each of the $k^2$ approximation function, takes $\ds O\big(
(k/\eps)^{2} n \big)$ time. Each one of them is a $\pth[]{\eps/k,
   O(k^2/\eps^2), O(k^2 n/\eps^2) }$-approx\-\si{imation}, which means
that the final arrangement is the overlay of $O(k^4/\eps^2)$ convex
polygons, each of complexity $O(k^2 n/\eps^2)$. In particular, any
pair of such polygons can have at most $O(k^2 n/\eps^2)$ intersection
points, and thus the overall complexity of the arrangement of these
polygons is
\begin{math}
    N = O \pth{ \pth{k^4/\eps^2}^2 (k^2n/\eps^2)} = O\pth{
       k^{10}\eps^{-6} n}.%
\end{math}
Computing this arrangement can be done by a standard sweeping
algorithm. Observing that every vertical line crosses only $O(k^4
/\eps^2)$ segments, imply that the sweeping can be done in $O\pth{
   \log (k/\eps) }$ time per operation, which implies that the overall
running time is
\begin{align*}
    O\pth{ k^2 \frac{k^2}{\eps^2} n + N \log \frac{k}{\eps}}
    \InESAVer{&}%
    =%
    O\pth{ \frac{k^{10}}{\eps^6} n \log \frac{k}{\eps} }.
\end{align*}

\paragraph{The result.}
\begin{theorem}
    \thmlab{main}%
    Given two simple polygons $\Poly$ and $\PolyA$ of total complexity
    $n$, one can compute a translation which $\eps$-approximates the
    maximum area of overlap of $\Poly$ and $\PolyA$. The time required
    is $O(c' n)$ where $\ds \Biggl. c' = \frac{k^10}{\eps^6}\log
    \frac{k}{\eps}$, where $k$ is the minimum number of convex
    polygons in the decomposition of $\Poly$ and $\PolyA$.

    More specifically, one gets a data-structure, such that for any
    query translation $\trans$, one can compute, in $O( \log n)$ time,
    an approximation $\aFunc(\trans)$, such that
    $\cardin{\aFunc(\trans) - \AreaOverlap{\Poly}{\PolyA}} \leq \eps
    \maxOverlap{ \Poly}{\PolyA}$, where $\maxOverlap{ \Poly}{\PolyA}$
    is the maximum area of overlap between $\Poly$ and $\PolyA$.
\end{theorem}

Note, that our analysis is far from tight. Specifically, for the sake
of simplicity of exposition, it is loose in several places as far as
the dependency on $k$ and $\eps$.


\subsection*{Acknowledgments}

The authors would like to thank the anonymous referees for their
insightful comments. In particular, the improved construction of
\lemref{approx:small:large} was suggested by an anonymous referee.

\newcommand{\etalchar}[1]{$^{#1}$}
 \providecommand{\CNFX}[1]{ {\em{\textrm{(#1)}}}}%
  \providecommand{\CNFSoCG}{\CNFX{SoCG}}%
  \providecommand{\CNFCCCG}{\CNFX{CCCG}}%
  \providecommand{\CNFFOCS}{\CNFX{FOCS}}%


\appendix

\section{Proof of \lemref{approx:scheme}}
\apndlab{fast:convex:approx}

\lemref{approx:scheme} follows by the work of Ahn \etal
\cite{acpsv-motpc-07}. For the sake of completeness we include the
details here (our construction is somewhat different).

\begin{proof}
    The width of $\Poly$ is realized by two parallel lines $\Line_0$
    and $\Line_{m+1}$ at a distance of at most $\WidthX{\Poly}$ from
    each other. These two lines can be computed in linear time using
    rotating caliper \cite{t-sgprc-83} (and one of them must pass
    through an edge of $\Poly$). Similarly, it is easy to compute the
    two extreme points of $\Poly$ in the direction of $\Line$. Adding
    the supporting lines to $\Poly$ through these points that are
    orthogonal to $\Line$, results in a rectangle, that has five
    vertices of $\Poly$ on its boundary, and we mark these vertices.
    Next, we slice this rectangle by $m$ parallel lines $\Line_1,
    \ldots, \Line_m$ into $m+1$ slices of the same width. We mark all
    the intersections of these lines with $\Poly$. Next, let $\Poly'$
    be the polygon formed by the convex-hull of all the marked
    points. It is easy to verify that $\Poly'$ can be computed in linear
    time.
    
    \parpic[r]{\includegraphics{\si{figs/approx_polygon}}}

    Clearly, $\Poly' \subseteq \Poly$. For any point $\pnt \in \Poly$,
    consider the two lines $\Line_i$ and $\Line_{i+1}$ that $\pnt$ is
    contained between them. Consider the projection of $\pnt$ into
    $\Line_i$ and $\Line_{i+1}$, denoted by $\pnt_i$ and $\pnt_{i+1}$
    respectively. If any of these two projections are inside $\Poly$,
    then we are done, as this portion of $\Poly$ is inside $\Poly'$,
    and the distance of projection is at most
    $\WidthX{\Poly}/m$. Thus, the only remaining possibility is that
    both projections are outside $\Poly$. But this implies that
    $\Poly$'s extreme point in the parallel direction to $\Line$ (or
    its reverse) must be in $\pnt$'s slice. In particular, the segment
    $\pnt_i \pnt_{i+1}$ (that includes $\pnt$) must intersect $\Poly$,
    and furthermore, since the aforementioned extreme point is a
    vertex of $\Poly'$, it follows that this segments must intersect
    $\Poly'$, implying the claim.
    \InESAVer{ \qed}
\end{proof}

\section{Counter example for some natural approaches}
\apndlab{counter:example}

\begin{quote}
    ``I have not failed. I've just found 10,000 ways that won't
    work.''  --― Thomas Edison
\end{quote}

Here we present a counterexample to the natural for our problem:
Decompose the two polygons into their convex parts, approximate the
parts, and then find the maximum overlap for the two reconstituted
polygons.

\begin{figure}
    \begin{tabular}{cc}
        \includegraphics[width=0.45\linewidth]{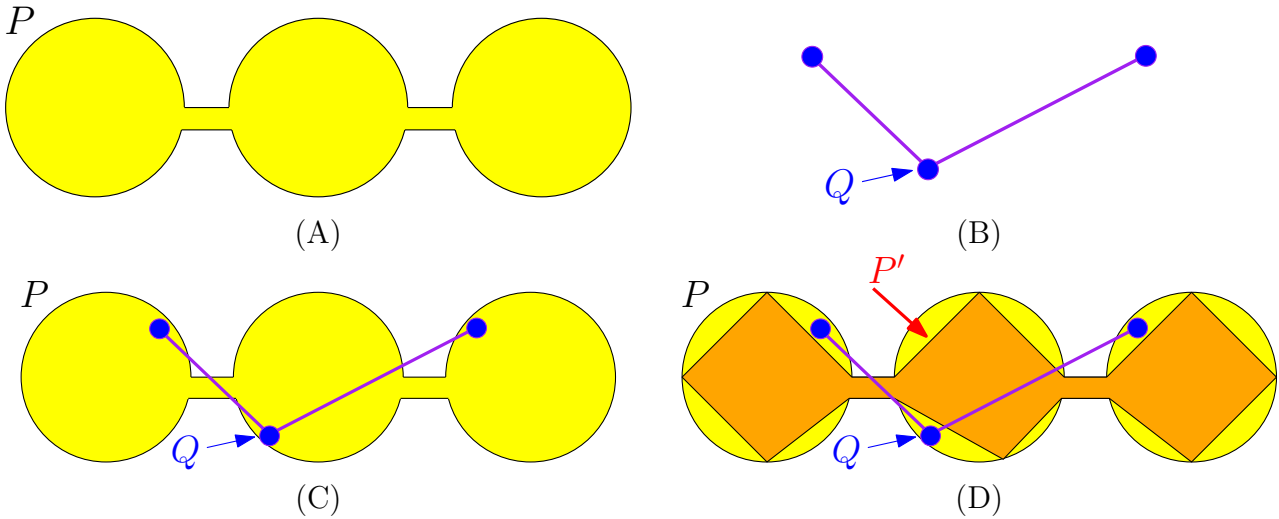}%
        &%
        \includegraphics[page=2,width=0.45\linewidth]{figs/no_good_2}%
        \\%
        (A) & (B) %
        \\%
        \includegraphics[page=3]{figs/no_good_2}%
        &%
        \includegraphics[page=4]{figs/no_good_2}%
        \\%
        (C) & (D) %
    \end{tabular}
    \caption{Counterexample to naive approach for approximating
       overlap. %
       (A), (B) Input polygons. (C) Optimal placement.  (D)
       Approximate polygons fails to preserve optimal solution.}
    \figlab{counter:example}
\end{figure}

So, consider the polygon $P$, depicted in \figref{counter:example}
(A), and its arch-nemesis $Q$ in \figref{counter:example} (B).  The
optimal overlap placement is depicted in \figref{counter:example} (C).

But any naive piecewise approximation of $P$ by pieces will fail,
because it creates cavities that might contain the whole overlap. For
example, one possible solution for the approximate polygon $P'$ and
the maximum overlap, is depicted in \figref{counter:example} (D).

In particular, any placement that makes the approximate polygon $P'$
cover two of the disks of $Q$, would fail to cover the third disk of
$Q$. so, in this case, the best approximation suggested by the
reviewer is (roughly) $2/3$. This can be made to be arbitrarily bad by
enlarging $k$.


\paragraph{Another approach that does not work.}

Another natural idea is to pick from each pair of polygons $P_i$ and
$Q_j$ a set of directions such that if you approximate each pair with
these directions than the pair overlap is approximated
correctly. Note, however, that if the polygons are of completely
different sizes then, inherently, you need unbounded number of
directions if the maximum overlap of the original (non-convex
polygons) align the pieces far from their piecewise overlap maximum.

\InsertAppendixOfProofs

\end{document}